\documentclass[letterpaper, 10 pt, conference]{ieeeconf}  

\IEEEoverridecommandlockouts                              
                                                          
\usepackage{graphicx} 
\usepackage{times} 
\usepackage{amsmath,amsfonts} 
\usepackage{amssymb}  
\usepackage{physics}
\usepackage{caption}
\usepackage{subcaption}

\usepackage{amsmath} 
 
\usepackage{amsthm}
\usepackage{amssymb}

\usepackage{xcolor}
\newtheorem{thm}{Theorem}
\newtheorem{lem}[thm]{Lemma}
\newtheorem{defi}[thm]{Definition}

\allowdisplaybreaks
\title{\LARGE \bf
Event-Triggered Boundary Control of Mixed-Autonomy Traffic}

\author{Yihuai Zhang, Huan Yu$^{*}$%
\thanks{Yihuai Zhang and Huan Yu are with the Hong Kong University of Science and Technology (Guangzhou), Thrust of Intelligent Transportation, Guangzhou, Guangdong, China. Huan Yu is also affiliated with the Hong Kong University of Science and Technology, Department of Civil and Environmental Engineering, Hong Kong SAR, China.}%
\thanks{$^{*}$ corresponding author (huanyu@ust.hk). }}

\begin{document}
\bibliographystyle{abbrv}

\maketitle
\thispagestyle{empty}
\pagestyle{empty}

\begin{abstract}
Control problems of mixed-autonomy traffic systems that consist of both human-driven vehicles (HV) and autonomous vehicles (AV), have gained increasing attention. This paper focuses on suppressing traffic oscillations in the mixed-autonomy traffic system using boundary control design. The mixed traffic dynamics are described by $4\times 4$ hyperbolic partial differential equations (PDEs), governing the propagation of four waves of traffic, including the density of HV, the density of AV, the friction between the two vehicle classes from driving interactions and the averaged velocity. We propose an event-triggered boundary control design since control signals of the traffic light on ramp or the varying speed limit cannot be continuously updated. We apply the event-triggered mechanism for a PDE backstepping controller and obtain a dynamic triggering condition. Lyapunov analysis is performed to prove the exponential stability of the closed-loop system with the event-triggered controller. Numerical simulation demonstrates the efficiency of the proposed event-trigger control design. We analyzed how the car-following spacing of AV affects the event-triggering mechanism of the control input in mixed-autonomy traffic.

\end{abstract}
\section{Introduction}
Various boundary control designs have been studied for the suppression of freeway traffic congestion~\cite{siri2021freeway}. The practical implementation of boundary control signals is achieved through traffic lights on ramps and variable speed limits (VSL). With the rapid development of autonomous driving technology, the penetration of autonomous vehicles (AV) in traffic has increased over the years, resulting in a mixed-autonomy traffic system consisting of both human-driven vehicles (HV) and AV. Traffic oscillations may arise from car-following and lane-changing interactions between AV and HV. The design of boundary control strategies for mixed-autonomy traffic remains an open question.

Event-based control is a computer control strategy aimed at improving system efficiency by updating the controller aperiodically. It requires defining a triggering condition for determining the time instant at which the controller needs to be updated. The triggering condition can be static or dynamic~\cite{girard2014dynamic,tabuada2007event,wang2024event}. Event-triggered control(ETC) for hyperbolic PDEs was first developed by Espitia~\cite{espitia2016event}. Initially, a continuous control input that can stabilize the system is designed, with an event-triggered mechanism embedded within the system. The event-triggered mechanism operates based on the system states and the state that describes the dynamics of the mechanism. The event-triggered control input, which also stabilizes the system, is then obtained through this mechanism. In traffic applications, Yu~\cite{yu2019traffic} first presented results for stabilizing traffic oscillations using the backstepping method, providing both theoretical guarantees and practical application potential. The author in~\cite{espitia2022traffic} developed an event-triggered output-feedback controller for cascaded roads. Event-triggered control has garnered interest due to its efficient use of communication and computational resources by updating the control law periodically.

The traffic dynamics for pure HV traffic are described by hyperbolic PDEs, such as the second-order Aw-Rascle-Zhang (ARZ) model~\cite{aw2000resurrection,zhang2002non}. Backstepping boundary control has been developed for the ARZ model~\cite{yu2019traffic,yu2022simultaneous,espitia2022traffic}. In addition to the backstepping control method, feedback control~\cite{karafyllis2019feedback,zhang2017stochastic} and optimal control~\cite{goatin2016speed,gomes2006optimal} can also be applied for boundary stabilization of traffic PDE models. Boundary control input is implemented by manipulating the red-green phase of traffic lights on ramps and the velocity display of VSL. Continuous control input needs to be updated periodically at each time step, which poses challenges for practical implementation. Some studies have developed ramp metering controllers for discrete-time traffic systems and designed discrete control laws. In~\cite{pasquale2020hierarchical}, a hierarchical centralized/decentralized event-triggered control method was proposed to reduce computation and communication load. Ferrara~\cite{ferrara2015event} introduced event-triggered model predictive schemes for discrete models of freeway traffic control. However, previous efforts to mitigate traffic congestion have primarily focused on traffic consisting solely of HV. For mixed-autonomy traffic, the interactive driving behaviors of AV and HV make boundary control problems more complex. In this paper, we propose an event-triggered control method for the mixed-autonomy traffic system.

The main contributions of this paper are twofold: we propose the first event-triggered controller for a mixed-autonomy traffic system modeled by an extended ARZ model and provide a theoretical guarantee through Lyapunov analysis. For application relevance, the results can be applied to traffic management systems to reduce computational resources and improve the overall efficiency of traffic operations. This work paves the way for deployment of advanced traffic management strategies.

The paper is organized as follows. Section~\ref{sec2} introduces the mixed-autonomy traffic system using the extended ARZ model. In Section~\ref{sec3}, the boundary control model is derived, and the backstepping controller in its continuous form is proposed. In Section~\ref{sec4}, the event-triggered boundary controller is developed, and Lyapunov analysis is performed to the closed-loop system. Section~\ref{sec5} presents the numerical simulation results, while Section~\ref{sec6} concludes the paper.

\section{Mixed-autonomy Traffic PDE Model}\label{sec2}
Motivated by the two-class extended ARZ PDE model~\cite{burkhardt2021stop}, the mixed-autonomy traffic consisting of HV and AV is proposed as
\begin{align}
    \partial_t \rho_{\rm h}+\partial_x\left(\rho_{\rm h} v_{\rm h}\right) &=0, \\
    \partial_t\left(v_{\rm h}-V_{e,\rm h}\right)+v_{\rm h} \partial_x\left(v_{\rm h} -V_{e,\rm h} \right)&=\frac{V_{e, \rm h}-v_{\rm h}}{\tau_{\rm h}}, \\
    \partial_t \rho_{\rm a}+\partial_x\left(\rho_{\rm a} v_{\rm a}\right) &=0,\\
    \partial_t\left(v_{\rm a}-V_{e, \rm a}\right)+v_{\rm a} \partial_x \left(v_{\rm a}-V_{e, \rm a}\right) &=\frac{V_{e, \rm a}-v_{\rm a}}{\tau_{\rm a}},
\end{align}
${\rho}_{\rm h}(x,t)$ and ${\rho}_{\rm a}(x,t)$ are the traffic densities of HV and AV, ${v}_{\rm h}(x,t)$, ${v}_{\rm a}(x,t)$ are the traffic velocities of HV and AV. The boundary conditions are set as
\begin{align}
    \rho_{\rm h}(0, t)  =\rho_{\rm h}^\star,
    \rho_{\rm a}(0, t) & =\rho_{\rm a}^\star, \\
    \rho_{\rm h}(0, t) v_{\rm h}(0, t)+\rho_{\rm a}(0, t) v_{\rm a}(0, t) & =\rho_{\rm h}^\star v_{\rm h}^\star+\rho_{\rm a}^\star v_{\rm a}^\star, \\
    \rho_{\rm h}(L, t) v_{\rm h}(L, t)+\rho_{\rm a}(L, t) v_{\rm a}(L, t) & =q_{\rm h}^\star+q_{\rm a}^\star + U(t),
\end{align}
where the spatial and time domain is defined as $(x,t) \in [0,L] \times \mathbb{R}^+$,  $\rho_{\rm h}^\star$, $\rho_{\rm a}^\star$ are the equilibrium densities, and $v_{\rm h}^\star$, $v_{\rm a}^\star$ are the equilibrium speeds. We will design event-triggered boundary controller  boundary control signal of ramp metering or VSL. We define the area occupancy $AO$ to describe the interaction between the two-class vehicles on the road~\cite{burkhardt2021stop, mohan2017heterogeneous}
\begin{align}
    AO(\rho_{\rm h},\rho_{\rm a}) = \frac{a_{\rm h} \rho_{\rm h} + a_{\rm a} \rho_{\rm a}}{W},
\end{align}
where $W$ is the road width. The impact area for HV $a_{\rm h}$ and AV $a_{\rm a}$ can be described as: 
\begin{align}
    a_{\rm h} = d \times (l + s_{\rm h})a_{\rm a} = d \times (l + s_{\rm a})
\end{align}
where $d$ is the vehicle width, $l$ is the vehicle length. We assume that the width and length are the same. $s_{\rm h}$ is the car-following gap of HV, $s_{\rm a}$ is the car-following gap of AV. The fundamental diagram based on the area occupancy is introduced for velocity-density equilibrium relation as:
\begin{align}
    v_{\rm h}^\star = V_{e,\rm h}(\rho_{\rm h},\rho_{\rm a}) = V_{\rm h} \left( 1 - \left(\frac{AO}{\overline{AO}_{\rm h}}\right)^{\gamma_{\rm h}}\right),\\
    v_{\rm a}^\star = V_{e,\rm a}(\rho_{\rm h},\rho_{\rm a}) = V_{\rm a} \left( 1 - \left(\frac{AO}{\overline{AO}_{\rm a}}\right)
    ^{\gamma_{\rm a}}\right),
\end{align}
where $V_{\rm h}$, $V_{\rm a}$ are the maximum speed, $\overline{AO}_{\rm h}$, $\overline{AO}_{\rm a}$ are the maximum area occupancy, $\gamma_{\rm h}$, $\gamma_{\rm a}$ are the traffic pressure exponent.

Compared to HV, AV tends to have a larger spacing due to the conservative driving strategies they have equipped. A larger spacing leads to a larger impact area, inducing the "creeping effect" on the road that HV takes over AV in congested regimes.

\section{Backstepping Control Design}\label{sec3}
\subsection{Boundary control model}
Linearizing the system at its equilibrium point $\rho_{\rm h}^\star$, $\rho_{\rm a}^\star$, $v_{\rm h}^\star$, $v_{\rm a}^\star$ and defining a small deviation $\Tilde{\rho}_{\rm h}(x,t) = \rho_{\rm h}(x,t) - \rho_{\rm h}^\star$, $\Tilde{v}_{\rm h}(x,t) = v_{\rm h}(x,t) - v_{\rm h}^\star$, $\Tilde{\rho}_{\rm a}(x,t) = \rho_{\rm a}(x,t) - \rho_{\rm a}^\star$, $\Tilde{v}_{\rm a}(x,t) = v_{\rm a}(x,t) - v_{\rm a}^\star$. Writing the system in an augmented expression $\mathbf{z}(x,t) = \begin{bmatrix}
    \Tilde{\rho}_{\rm h}(x,t)  &  \Tilde{v}_{\rm h}(x,t) & \Tilde{\rho}_{\rm a}(x,t) & \Tilde{v}_{\rm a}(x,t)
\end{bmatrix}^\mathsf{T}$ Defining the matrix $\mathbf{V} =\{ \Hat{v}_{ij}\}_{1 \leq i, j \leq 4}$ such that the coefficient matrix is diagonalized as $\mathbf{V}^{-1} \mathbf{J}_{\lambda} \mathbf{V}$ = $\text{Diag}\{\lambda_1,\lambda_2,\lambda_3,\lambda_4\}$, with  positive eigenvalues in ascending order. We also define the source term matrix as $\Hat{\mathbf{J}} = \mathbf{V}^{-1}{\mathbf{J}} \mathbf{V} = \{ \Hat{J}_{ij} \}_{1 \leq i,j \leq 4}$. 
The transformation matrix $\mathbf{T} $ is given as
\begin{align}
     \mathbf{T}  = \begin{bmatrix}
        \\
       \mathbf{T} ^{+}\\
       \\ \hline 
       \mathbf{T} ^{-}
   \end{bmatrix}= \begin{bmatrix}
       0 & \mathrm{e}^{-\frac{\Hat{J} _{22}}{v_{\rm a}^*} x} & 0 & 0 \\
0 & 0 & \mathrm{e}^{-\frac{\Hat{J}_{33} }{\lambda_{3}} x} & 0 \\
\mathrm{e}^{-\frac{\Hat{J}_{11} }{v_{\rm h}^*} x} & 0 & 0 & 0 \\
0 & 0 & 0 & \mathrm{e}^{-\frac{\Hat{J}_{44} }{\lambda_{4}} x}
\end{bmatrix} \mathbf{V} ^{-1},
\end{align}
where $\mathbf{T}^{+} \in \mathbb{R}^{3\times 4}$ and $\mathbf{T}^{-} \in \mathbb{R}^{1\times 4}$. The change of coordinates is
\begin{align}
\begin{bmatrix}
    w_1 & w_2 & w_3 & w_4
\end{bmatrix}^\mathsf{T}=\mathbf{T} \mathbf{z}. \label{transfomation}
\end{align}
Then we perform Riemann transformation of the linarized system, thus we get 
\begin{align}
        \mathbf{w}^+_t(x,t) +\Lambda^{+} \mathbf{w}^+_x(x,t) =& \Sigma^{++}(x)\mathbf{w}^+ (x,t)\nonumber\\
        &+\Sigma^{+-}(x) \mathbf{w}^-(x, t), \label{clpsys1}\\
        \mathbf{w}^-_t(x, t)-\Lambda^{-} \mathbf{w}^-_x(x, t) =& \Sigma^{-+}(x) \mathbf{w}^+(x,t),\label{clpsys2}\\
    \mathbf{w}^+(0,t)  =& {Q} \mathbf{w}^-(0, t), \label{clpsys3} \\
    \mathbf{w}^-(L, t) =& {R}\mathbf{w}^+(L, t)+\bar{U}(t), \label{clpsys4}
\end{align}
where $\mathbf{w}^{+} = [w_1, w_2, w_3]^\mathsf{T}$, $\mathbf{w}^{-} = w_4$. The coefficient matrices are given as
$\Lambda^{+} =\text{Diag}\{\lambda_{1},\lambda_{2},\lambda_{3}\}$, $\Lambda^- = -\lambda_4$, $\Sigma^{++}(x)$, $\Sigma^{+-}(x)$, $\Sigma^{-+}(x)$, $Q\in \mathbb{R}^{3 \times 1}$, and $R\in \mathbb{R}^{1\times 3}$ coefficients which can be obtained by the Riemann transformation. The details of the coefficients can be found in~\cite{burkhardt2021stop}. Also, $\bar{U}(t)=\mathrm{e}^{-\frac{\Hat{J}_{44} }{\lambda_{4}} L} \frac{1}{\kappa} U(t)$, $\kappa=v_{\rm h}^* \Hat{v}_{1 4} +\rho_{\rm h}^* \Hat{v}_{2 4} +v_{\rm a}^* \Hat{v}_{3 4} +\rho_{\rm a}^* \Hat{v}_{4 4}$.
The eigenvalues were shown to satisfy the following condition: \cite{zhang2006hyperbolicity}
\begin{align}
     \lambda_{4} \leq \min\{ \lambda_{1},\lambda_{3} \} \leq \lambda_{2}\leq\max\{\lambda_{1},\lambda_{3}\}.
\end{align}
The traffic system would be congested if $\lambda_4 < 0$.  In the congested regime, the traffic information propagates from downstream to upstream and the efficiency of the traffic system becomes low. The traffic system can be divided into free and congested regimes based on the direction of propagation of traffic waves.

\subsection{Backstepping transformation and controller design}
We consider the stabilization of the closed-loop system \eqref{clpsys1}-\eqref{clpsys4} with continuous control input at each time step. Defining the backstepping transformation:
\begin{align}
    \mathcal{K}\mathbf{w}= \begin{pmatrix} \mathbf{w}^{+} \\
    \mathbf{w}^- -\int_0^x \mathbf{K}(x, \xi)\mathbf{w}^{+}(\xi,t) +M(x,\xi) \mathbf{w}^-(\xi, t) d \xi \end{pmatrix},\label{back4}
\end{align}
where $\mathbf{w} = [\mathbf{w}^{+}, \mathbf{w}^{-}]$ and the backstepping control kernel $\mathbf{K}(x,\xi)\in \mathbb{R}^{1\times 3}$, $M(x,\xi) \in \mathbb{R}^{1\times 1}$ are defined as:
\begin{align}
    \mathbf{K}(x, \xi)=\begin{bmatrix}
        k_{1}(x, \xi) & k_{2}(x, \xi) & k_{3}(x, \xi)
    \end{bmatrix}
\end{align}
Both kernels are defined in the triangular domain $\mathcal{T}=\{0 \leq \xi \leq x \leq L\}$. And the target perturbed system is:
\begin{align}
\alpha_t(x,t)+\Lambda^{+} \alpha_x(x,t) =&\Sigma^{++}(x)\alpha(x,t) +\Sigma^{+-}(x) \beta(x,t) \nonumber \\
+\int_0^x\mathbf{C}^+(x,\xi)\alpha(\xi,t) d\xi &+ \int_0^x\mathbf{C}^-(x,\xi)\beta(\xi,t)d\xi,
\label{tarsys1}\\
\beta_t\left(x, t\right)-\Lambda^{-} \beta_x(x, t)&=0, \label{tarsys2}\\
    \alpha(0,t)  &= {Q} \beta(0, t) \label{tarsys3} \\
    \beta(L, t) &= 0 \label{tarsys4}
\end{align}
where $\alpha = [\alpha_1, \alpha_2, \alpha_3]^\mathsf{T}$. The coefficients $ \mathbf{C}^+(x,\xi)\in \mathbb{R}^{3\times 3}$ and $\mathbf{C}^-(x,\xi)\in \mathbb{R}^{3\times1}$ are defined in the same triangular domain $\mathcal{T}$. 
The kernel equations are stated in~\cite{burkhardt2021stop} and the well-posedness of the target system and the kernel equations are proved in \cite{hu2015control,zhang2024mean}. The control input is given as:
\begin{align}
     \Bar{U}(t)=& \int_0^L\left(\mathbf{K}(L, \xi)\mathbf{w}^+(\xi,t)  +M(L, \xi) \mathbf{w}^- (\xi, t)\right) d \xi. \nonumber\\
     &-{R} \mathbf{w}^+(L,t).\label{control_law_w}
\end{align}
\subsection{Inverse Transformation}
The transformation~\eqref{back4} is invertible such that the target system share the same properties with the original system. The inverse transformation turn the target system~\eqref{tarsys1}-\eqref{tarsys4} into the original system~\eqref{clpsys1}-\eqref{clpsys4}:
\begin{align}
    \mathcal{L}\vartheta  = \begin{pmatrix}
        \alpha\\
        \beta - \int_0^x (\mathbf{L}(x,\xi) \alpha(\xi,t) + N(x,\xi)\beta(\xi,t))d\xi
    \end{pmatrix}\label{invback}
\end{align}
where $\vartheta = [\alpha_1,\alpha_2,\alpha_3, \beta]^\mathsf{T}$ and $\mathbf{L}(x,\xi) \in \mathbb{R}^{3\times 1}$, $N(x,\xi)\in \mathbb{R}^{1\times 1}$ are defined as:
\begin{align}
    \mathbf{L}(x,\xi) = \begin{bmatrix}
        \ell_1(x,\xi) & \ell_2(x,\xi) & \ell_3(x,\xi)
    \end{bmatrix}
\end{align}
Inverse kernels are also defined in the same triangular domain $\mathcal{T}$. The inverse kernel equations are easily obtained in~\cite{hu2015control}.
The states $\mathbf{w}$ and $\vartheta$ have equivalent $L^2$ norm, i.e. there exist two constants $p_1>0$ and $p_2>0$ such that $p_1\norm{\mathbf{w}}_{L^2}^2 \leq \norm{\vartheta}_{L^2}^2\leq p_2\norm{\mathbf{w}}_{L^2}^2$, where $\vartheta = (\alpha_1,\alpha_2,\alpha_3,\beta)$. The continuous-time control input $\Bar{U}(t)$ is calculated using states $(\alpha,\beta)$ of target system:
\begin{align}
    \bar{U}(t)=&\int_0^L (\mathbf{L}(L, \xi)\alpha(\xi,t) +N(L, \xi) \beta (\xi, t)) {d} \xi \nonumber\\ 
    &-R\mathbf{w}^+(L,t)
    \label{control law}
\end{align}

\section{Event-triggered Boundary Control}\label{sec4}
In this section, we introduce the event-triggered conditions for the traffic system, which determine the time intervals at which the controller should be updated. We then ensure the exponential stability of the closed-loop system. First, we consider the stabilization of the closed-loop system based on events by sampling the continuous-time controller $\Bar{U}(t)$ at a certain sequence of time instants. The controller is updated when the triggering conditions are met. We then redefine the boundary control input in~\eqref{clpsys4}:
\begin{align}
    \mathbf{w}^-(L, t) = {R}\mathbf{w}^+(L, t)+\bar{U}_d(t),
\end{align}
where $\bar{U}d(t) = \bar{U}(t) + d(t)$ for all $t \in [t_k, t{k+1})$, $k \geq 0$. Here, $d(t)$ represents the deviation between the theoretical control input and the event-triggered control input. Consequently, the sampled control law is expressed as:
\begin{align}
     \bar{U}_d(t) = &\int_0^L (\mathbf{L}(L, \xi)\alpha(\xi,t_k) +N(L, \xi) \beta (\xi, t_k)) {d} \xi \nonumber\\
     &-R\mathbf{w}^+(L,t_k), \label{control law sam}
\end{align}
thus we get the actuation deviation $d(t)$:
\begin{align}
    d(t) = &-R(\mathbf{w}^+(L,t_k) - \mathbf{w}^+(L,t)) \nonumber\\ 
    &+ \int_0^L \bigg(\mathbf{L}(L, \xi)(\alpha(\xi,t_k)-\alpha(\xi,t))\nonumber\\
    &+N(L, \xi) (\beta (\xi, t_k)- \beta (\xi, t)\bigg) {d} \xi. \label{dt}
\end{align}
Applying the sampled control law $\Bar{U}_d(t)$ to the system \eqref{clpsys1}-\eqref{clpsys4}, we get the perturbed target system:
\begin{align}
\alpha_t(x,t)+\Lambda^{+} \alpha_x(x,t) =&\Sigma^{++}(x)\alpha(x,t) +\Sigma^{+-}(x) \beta(x,t) \nonumber \\
+\int_0^x\mathbf{C}^+(x,\xi)\alpha(\xi,t) d\xi &+ \int_0^x\mathbf{C}^-(x,\xi)\beta(\xi,t)d\xi,
\label{tarpd1}\\
\beta_t\left(x, t\right)-\Lambda^{-} \beta_x(x, t)&=0, \label{tarpd2}\\
\alpha(0,t)  &= {Q} \beta(0, t), \label{tarpd3} \\
\beta(L, t) &= d(t). \label{tarpd4}
\end{align}
We consider a triggering condition relies on the evolution of $d(t)$ and the following Lyapunov function,
\begin{align}
     V(t) = \int_0^L \sum_{i=1}^3 \frac{A_i}{\lambda_i}\mathrm{e}^{-\frac{\mu x}{\lambda_i}} \alpha_i^2(x,t) + \frac{B}{\Lambda^-} \mathrm{e}^{\frac{\mu x}{\Lambda^-}} \beta^2(x,t) dx,
     \label{lyapunov}
\end{align}
where the constant coefficients $A_1$, $A_2$, $A_3$, $B$ and $\mu$ are positive. The Lyapunov candidate is equivalent to the $L_2$ norm of the state $\vartheta$, therefore, there exist two constants $p_3>0$ and $p_4>0$ such that $p_3\norm{\vartheta}_{L^2}^2 \leq V(t) \leq p_4\norm{\vartheta}_{L^2}^2$.

\subsection{Dynamic triggering condition}
We define the event-triggered mechanism (ETM) using the dynamic triggering condition which can be derived by the evolution of the controller deviation \eqref{dt} and another dynamic variable $m(t)$.
\begin{defi}
    Let the Lyapunov candidate $V(t)$ be given by \eqref{lyapunov}. The event-triggered controller is defined in \eqref{control law sam} with a dynamic event-triggered mechanism. The time of the execution $t_k \geq 0$ from $t_0 = 0$ in a finite number set of times. The set is determined by:
    \begin{itemize}
        \item \text{if} \{$t>t_k$ $\wedge $ $\zeta B \mathrm{e}^{\frac{\mu L}{\Lambda^-}}d^2(t)$ $\geq$ $ \zeta \mu \sigma V(t) - m(t)$\} = $\emptyset$, then the set of the times of the events is \{ $t_0,\dots,t_k$\}.
        \item  \text{if} \{$t>t_k$ $\wedge $ $ \zeta B \mathrm{e}^{\frac{\mu L}{\Lambda^-}}d^2(t)$ $\geq$ $ \zeta \mu \sigma V(t) - m(t)$\} $\ne$ $\emptyset$, then the next execution time is determined by:
        $t_{k+1}$ $=$ $\inf\{ t>t_{k} \wedge  B \mathrm{e}^{\frac{\mu L}{\Lambda^-}}d^2(t)$ $\geq  \zeta \mu \sigma V(t) - m(t)\}$,
    \end{itemize}
    where $m(t)$ satisfies the ordinary differential equation,
    \begin{align}
        \Dot{m}(t) = &- \eta m(t) +  B \mathrm{e}^{\frac{\mu L}{\Lambda^-}}d^2(t) - \sigma \mu V(t) \nonumber\\
        &- \sum_{i=1}^3 \varsigma_i\alpha_i^2(L,t) - \varsigma_4 \beta^2(0,t),
        \label{mdyna}
    \end{align}
    where $\zeta>0$, $\mu>0$, $\sigma > 0$, $\varsigma_i>0$, $i\in \{1,2,3,4\}$, $\eta >0$ and $m(0) = m^0$.
    \label{defdynamic}
\end{defi}
Based on the definition \ref{defdynamic}, we have the following result for $m(t)$.
\begin{lem}
    Under the ETM in Definition \ref{defdynamic}, it holds that $\zeta B \mathrm{e}^{\frac{\mu L}{\Lambda^-}}d^2(t)$ $-$ $ \zeta \mu \sigma V(t) + m(t)$ $\leq 0$ with $m(t) \leq 0$
    \label{mbound}
\end{lem}
\begin{proof}
    We already know the ETM in Definition \ref{defdynamic}. It holds that the system in the simulation period always guarantee the following condition,
    \begin{align}
        \zeta B \mathrm{e}^{\frac{\mu L}{\Lambda^-}}d^2(t) -  \zeta \mu \sigma V(t) \leq - m(t).
    \end{align}
    We have the result
    \begin{align}
         B \mathrm{e}^{\frac{\mu L}{\Lambda^-}}d^2(t) -  \mu \sigma V(t) \leq -\frac{1}{\zeta} m(t),
    \end{align}
    using \eqref{mdyna}, we get 
    \begin{align}
        \Dot{m}(t) \leq -\eta m - \frac{1}{\zeta} m(t) -\sum_{i=1}^3 \varsigma_i \alpha_i^2(L,t) - \varsigma_4\beta^2(0,t).
    \end{align}
    Using the comparison principle, we have 
    \begin{align}
        m(t) \leq 0 , \forall t \geq 0,
    \end{align}
    this finishes the proof of Lemma \ref{mbound}.
\end{proof}
We also have the following lemma for the bound of the actuation deviation $d(t)$.
\begin{lem}
    There exists $\epsilon_i >0, i\in\{1,2,3\}$, $\phi_1$ and $\phi_2 > 0$, for the $d(t)$ introduced in \eqref{dt} with $t \in (t_k, t_{k+1})$, such that
    \begin{align}
        \Dot{d}^2(t) \leq \sum_{i=1}^3\epsilon_i \alpha^2_i(L,t) + \phi_1 d^2(t)+ \phi_2 V(t).
    \end{align}
    \label{bounddt}
\end{lem}
\begin{proof}
   Taking the time derivative of $d(t)$ and using the dynamics of the perturbed target system in \eqref{tarpd1}-\eqref{tarpd4} and integrating by parts, we get
    \begin{align}
        &\Dot{d}(t) = \mathbf{L}(L,L)\Lambda^+\alpha(L,t) - N(L,L)\Lambda^- \beta(L,t) \nonumber\\
        &+ \left( N(L,0)\Lambda^- - \mathbf{L}(L,0)\Lambda^+Q \right)\beta(0,t)\nonumber\\
        &- \int_0^L \mathbf{L}_{\xi}(L,\xi)\Lambda^+\alpha(\xi,t)d\xi + \int_0^L N_{\xi}(L,\xi)\Lambda^-\beta(\xi,t)d\xi \nonumber\\
        &-\int_0^L \mathbf{L}(L,\xi)\Sigma^{++}(\xi)\mathbf{w}^+(\xi,t)d\xi \nonumber\\
        &- \int_0^L \mathbf{L}(L,\xi)\Sigma^{+-}(\xi)\mathbf{w}^-(\xi,t)d\xi.
    \end{align}
    Taking the square of $\Dot{d}(t)$, combining Young's inequality and the Cauchy-Schwarz inequality, we have
    \begin{align}
       &\Dot{d}^2(t) \leq  8\sum_{i=1}^3 \ell_i^2(L,L)\lambda_i^2\alpha_i^2(L,t) \nonumber\\
       &+ 8 N^2(L,L)(\Lambda^-)^2 \beta^2(L,t) + \frac{8}{p_3}c_1V(t) + \frac{8}{p_1p_3}c_2 V(t)\nonumber\\
       &\leq 8\sum_{i=1}^3 \ell_i^2(L,L)\lambda_i^2\alpha_i^2(L,t) + 8 N^2(L,L)(\Lambda^-)^2 d^2(t)\nonumber\\
       &+ \frac{8}{p_3} \left(c_1 + \frac{c_2}{p_1}\right)V(t),
    \end{align}
    where 
    $c_1 $ $=$ $ \max\{\int_0^L$ $(\mathbf{L}_\xi(L,\xi)\Lambda^+)^2 d\xi,$ $ \int_0^L$ $(N_\xi(L,\xi)\Lambda^-)^2 d\xi \}$
    $c_2 $ $ =$ $ \max\{\int_0^L (\mathbf{L}(L,\xi) \Sigma^{++}(\xi))^2 d\xi,$ $ \int_0^L (\mathbf{L}(L,\xi) \Sigma^{+-}(\xi))^2 d\xi \}$.
    And thus we get $\epsilon_i= 8 \ell_i^2(L,L)\lambda_i, i \in \{1,2,3\}$, $\phi_1 = 8N^2(L,L)(\Lambda^-)^2$, $\phi_2 = \frac{8}{p_3} \left(c_1 + \frac{c_2}{p_1}\right)$. This concludes the proof of Lemma \ref{bounddt}.
\end{proof}
\subsection{Avoidance of Zeno phenomenon}
Under the dynamic event triggering condition, the Zeno phenomenon should be avoided. In this section, we prove that the dynamic event-triggering condition for the system \eqref{clpsys1}-\eqref{clpsys4} avoids the Zeno phenomenon. We have the following theorem.
\begin{thm}
    There exists a minimal dwell time $\tau^\star>0$ between two adjacent trigger times, $t_{k+1} - t_k \geq \tau^\star, k \geq 0$, under the dynamic trigger condition in Definition \ref{defdynamic} with parameters $\zeta$, $\mu$, $\sigma$, $\varsigma_i,i \in\{1,2,3,4\}$, $\eta$, $\epsilon_i, i \in \{1,2,3\}$. And the parameters satisfy:
    \begin{align}
        \varsigma_i &\geq \max\{\zeta B \mathrm{e}^{\frac{\mu L}{\Lambda^-}}\epsilon_i,\zeta\mu\epsilon_i, i \in \{1,2,3\}\},\\
        \varsigma_4 &\geq \max\{0, -2\zeta\mu(\sum A_iq_i^2 - B), i \in \{1,2,3\}.
    \end{align}
    \label{minimalthm}
\end{thm}
\begin{proof}
    We know from Definition~\ref{defdynamic} that for all time $t \geq 0$, all events are executed to guarantee the following: 
    \begin{align}
        \zeta B \mathrm{e}^{\frac{\mu L}{\Lambda^-}} d^2(t) \leq  \zeta \mu \sigma V(t) - m(t).
    \end{align}
    Then we define the following function:
    \begin{align}
        \Psi(t) = \frac{\zeta B \mathrm{e}^{\frac{\mu L}{\Lambda^-}}d^2(t) + \frac{1}{2}m(t)}{ \zeta \mu \sigma V - \frac{1}{2}m(t)}.
    \end{align}
    The function $d(t)$ and $V(t)$ are continuous in time interval $[t_{k},t_{k+1}]$, therefore, the function $\Psi(t)$ is also a continuous function in $[t_{k},t_{k+1}]$. We can derive that there exists $t'_k > t_k$ such that $\forall t \in [t'_{k},t_{k+1}]$, $\Psi(t) \in [0,1]$ using the intermediate value theorem.
    Using Young's inequality, we have:
    \begin{align}
         &\Dot{\Psi}(t) \leq \frac{\zeta B \mathrm{e}^{\frac{\mu L}{\Lambda^-}}d^2}{ \zeta \mu \sigma V - \frac{1}{2}m} + \frac{\zeta B \mathrm{e}^{\frac{\mu L}{\Lambda^-}}\Dot{d}^2}{ \zeta \mu \sigma V - \frac{1}{2}m} \nonumber\\
         & +\frac{\frac{1}{2}(-\eta m +B \mathrm{e}^{\frac{\mu L}{\Lambda^-}}d^2 - \sigma\mu V)}{\zeta \sigma \mu V - \frac{1}{2}m}\nonumber\\
         & + \frac{\frac{1}{2}\left( - \sum_{i=1}^3 \varsigma_i\alpha_i^2(L,t) - \varsigma_4 \beta^2(0,t)\right)}{ \zeta \mu \sigma V - \frac{1}{2}m} -\frac{\zeta \mu \Dot{V}\Psi}{ \zeta \mu \sigma V - \frac{1}{2}m} \nonumber\\
         & + \frac{\frac{1}{2}\left( -\eta m +B \mathrm{e}^{\frac{\mu L}{\Lambda^-}}d^2 - \sigma\mu V \right)}{ \zeta \mu \sigma V - \frac{1}{2}m}\Psi \nonumber\\
         & +\frac{\frac{1}{2}\left( - \sum_{i=1}^3 \varsigma_i\alpha_i^2(L,t) - \varsigma_4 \beta^2(0,t)\right)}{ \zeta \mu \sigma V - \frac{1}{2}m} \Psi.
    \end{align}
    As defined in~\eqref{lyapunov},  the time derivative of $V(t)$ can be obtained by integrating by parts and using boundary conditions for perturbed target system. Thus, the $\Dot{V}$ is given as: 
    \begin{align}
        \Dot{V} &\leq - \sum_{i=1}^3 A_i \mathrm{e}^{-\frac{\mu L}{\lambda_i}}\alpha_i^2(L,t) +  (\sum_{i=1}^3 A_i q_i^2 - B) \beta^2(0,t) \nonumber\\
        & + B\mathrm{e}^{\frac{\mu L}{\Lambda^-}}d^2(t)-(\mu-\gamma) V,
    \end{align}
    where $\gamma = \frac{2A}{p_3\min\{ \lambda_i\}}( \max_{x\in [0,L]} \norm{\Sigma^{++}(x)} $ $+ (1+\frac{1}{p_1})$ $ \max_{x\in [0,L]} \norm{\Sigma^{+-}(x)})$.
    \begin{align}
        \Dot{V} &\leq - \sum_{i=1}^3 \alpha_i^2(L,t) +  \beta^2(0,t)  + B\mathrm{e}^{\frac{\mu L}{\Lambda^-}}d^2(t)-(\mu-\gamma) V,
    \end{align}
    Choosing $\varsigma_i \geq \zeta B \mathrm{e}^{\frac{\mu L}{\Lambda^-}}\epsilon_i$ and $\varsigma_i \geq \zeta\mu\epsilon_i$, $i\in \{1,2,3\}$, $\varsigma_4 > 0$ and $\varsigma_4 +2\zeta\mu(\sum A_iq_i^2 - B) >0$, we get the following equation after simplifying:
    \begin{align}
        \Dot{\Psi}(t) \leq&  \left(\frac{-\zeta\mu + \frac{1}{2}}{\zeta}\right)\Psi^2 +\left( 1+\phi_1 + \frac{1}{2\zeta} \right.\nonumber\\
        &+ \left.\frac{-\zeta\mu \sigma + \frac{1}{2}}{\zeta} + \eta + \frac{\zeta\mu \sigma(\mu-\gamma) - \frac{1}{2}\mu \sigma}{\zeta\mu \sigma} \right)\Psi\nonumber\\
        &+\left( \frac{(\zeta B \mathrm{e}^{\frac{\mu L}{\Lambda^-}}\phi_2 - \frac{1}{2} \mu \sigma)}{\zeta \mu \sigma} + \eta + 1+\phi_1+\frac{1}{2\zeta}\right).
    \end{align}
    Thus, the $\Psi(t)$ has the form:
    \begin{align}
        \Psi(t) \leq \varphi_1\Psi^2(t) + \varphi_2 \Psi(t) + \varphi_3,
    \end{align}
    where 
    $\varphi_1 = \frac{1}{2\zeta} - \mu \sigma$,
    $\varphi = 1+ \phi_1 + \frac{1}{2\zeta} (1-\sigma)\mu -\gamma +\eta$,
    $\varphi_3 = \frac{B \mathrm{e}^{\frac{\mu L}{\Lambda^-}}\phi_2}{\mu \sigma} + 1+\eta + \phi_1$. 
    Using the comparison principle, we get the time from $\Psi(t_k') = 0$ to $\Psi(t_{k+1}) = 1$ is at least
    \begin{align}
        \tau^\star = \int_0^L \frac{1}{\varphi_1 s^2 + \varphi_2s +\varphi_3} ds.
    \end{align}
    Then, $t_{k+1} - t_{k} \geq t_{k+1} - t_{t'_k} = \tau^\star$. This finishes the proof of Theorem \ref{minimalthm}.
\end{proof}
Now we have proved that there exists a minimal dwell time between two adjacent events. The Zeno phenomenon is avoided. Based on the previous results, the exponential stability of the system \eqref{clpsys1}-\eqref{clpsys4} with the event-triggered controller \eqref{control law sam} was obtained, as stated in Theorem~\ref{esthm}.
\begin{thm}
    Let $A_i >0, i\in\{1,2,3\}$, $B>0$, $\zeta>0$, $\eta \in (0,1)$, $\varsigma_i,i\in\{1,2,3,4\} \in (0,1)$ such that
    \begin{align}
        \varsigma_i - A_i \mathrm{e}^{-\frac{\mu L}{\lambda_i}} &\leq 0 , i \in \{1,2,3\}\\
        \varsigma_4 + \sum_{i=1}^3 A_i q_i^2 - B &\leq 0, i \in \{1,2,3\}
    \end{align}
    $V$ is given by \eqref{lyapunov} and $d$ is given by \eqref{dt}. The system \eqref{clpsys1}-\eqref{clpsys4} with the event-triggered controller \eqref{control law sam} is exponentially stable under the ETM in Definition \ref{defdynamic}.
    \label{esthm}
\end{thm}
\begin{proof}
    We consider the following Lyapunov candidate for perturbed target system \eqref{tarpd1}-\eqref{tarpd4},
    \begin{align}
        V_{d}(t,m) = V(t) - m(t).
    \end{align}
    Taking time derivative of the Lyapunov candidate, we get:
    \begin{align}
        \Dot{V}_{d}(t,m) &\leq  B\mathrm{e}^{\frac{\mu L}{\Lambda^-}}d^2(t)-(\mu-\gamma) V - \Dot{m}(t)\nonumber\\
        & - \sum_{i=1}^3 A_i \mathrm{e}^{-\frac{\mu L}{\lambda_i}}\alpha_i^2(L,t) +  (\sum_{i=1}^3 A_i q_i^2 - B) \beta^2(0,t),    
    \end{align}
    taking into the expression of $\Dot{m}(t)$, we get:
    \begin{align}
        \Dot{V}_{d} \leq &-(\mu - \gamma) V + B\mathrm{e}^{\frac{\mu L}{\Lambda^-}}d^2(t) + \eta m + \mu \sigma V \nonumber\\
        &-B \mathrm{e}^{\frac{\mu L}{\Lambda^-}}d^2(t) + \sum_{i=1}^3 (\varsigma_i - A_i \mathrm{e}^{-\frac{\mu L}{\lambda_i}})\alpha_i^2(L,t) \nonumber\\
        &+ (\varsigma_4 + \sum_{i=1}^3 A_i q_i^2 - B) \beta^2(0,t).
    \end{align}
    Simplifying the equation, thus we have:
    \begin{align}
        \Dot{V}_{d} \leq -(\mu(1-\sigma) - \gamma){V}_{d} + (\eta - (\mu(1-\sigma) - \gamma))m.
    \end{align}
    Choosing $\eta - (\mu(1-\sigma) - \gamma) \geq 0$, we get $ \Dot{V}_{d} \leq -(\mu(1-\sigma) - \gamma){V}_{d}$.
    Using the comparison principle again and $m(0) = 0$, we have:
    \begin{align}
        \Dot{V}(t) \leq \mathrm{e}^{-(\mu(1-\sigma) - \gamma)t} V(0).
    \end{align}
    Therefore, an estimation of the original system of the $L_2$ norm can be written as
    \begin{align}
        \norm{\mathbf{w}(x,t)}^2_{L^2} \leq \frac{p_2p_4}{p_1 p_3} \mathrm{e}^{-(\mu(1-\sigma) - \gamma)t}\norm{\mathbf{w}(0,t)}^2_{L^2}
    \end{align}
    This concludes the proof of Theorem \ref{esthm}.
\end{proof}

\section{Numerical Simulation}\label{sec5}
\begin{figure}[tbp]
    \centering
    \subcaptionbox{HV \label{HVsevent_s16}}{ \includegraphics[width =0.45\linewidth]{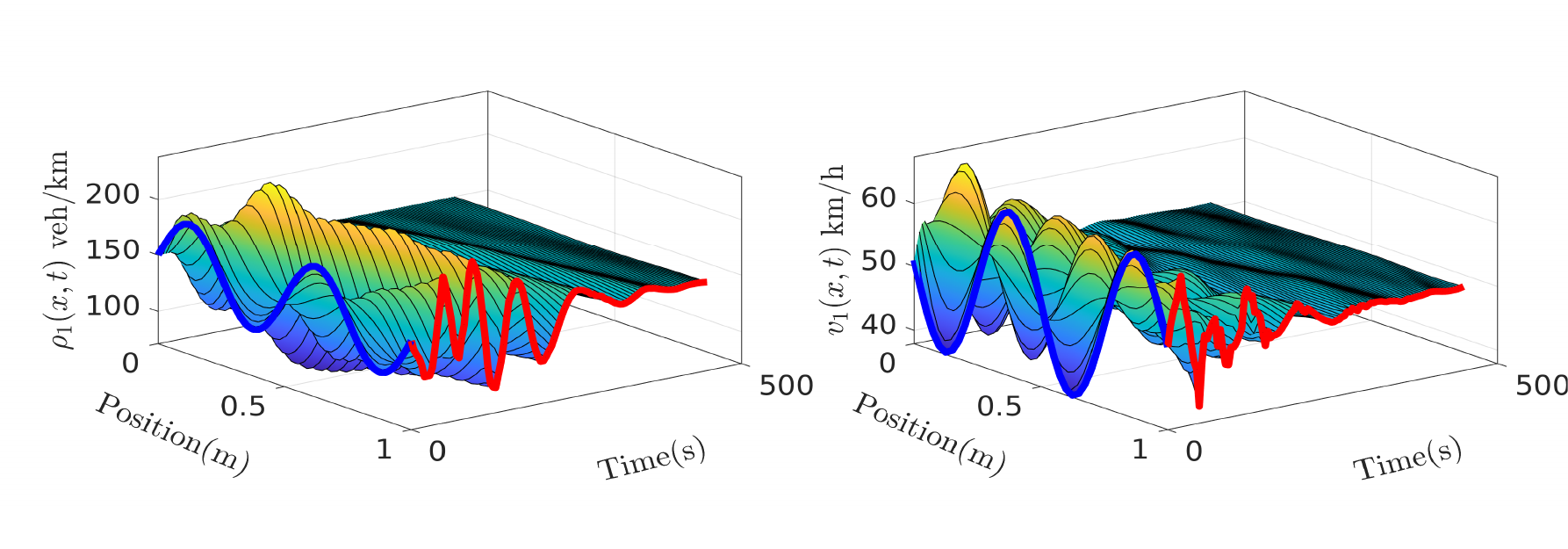}} 
    \subcaptionbox{AV \label{AVsevent_s16}}{\includegraphics[width =0.45\linewidth]{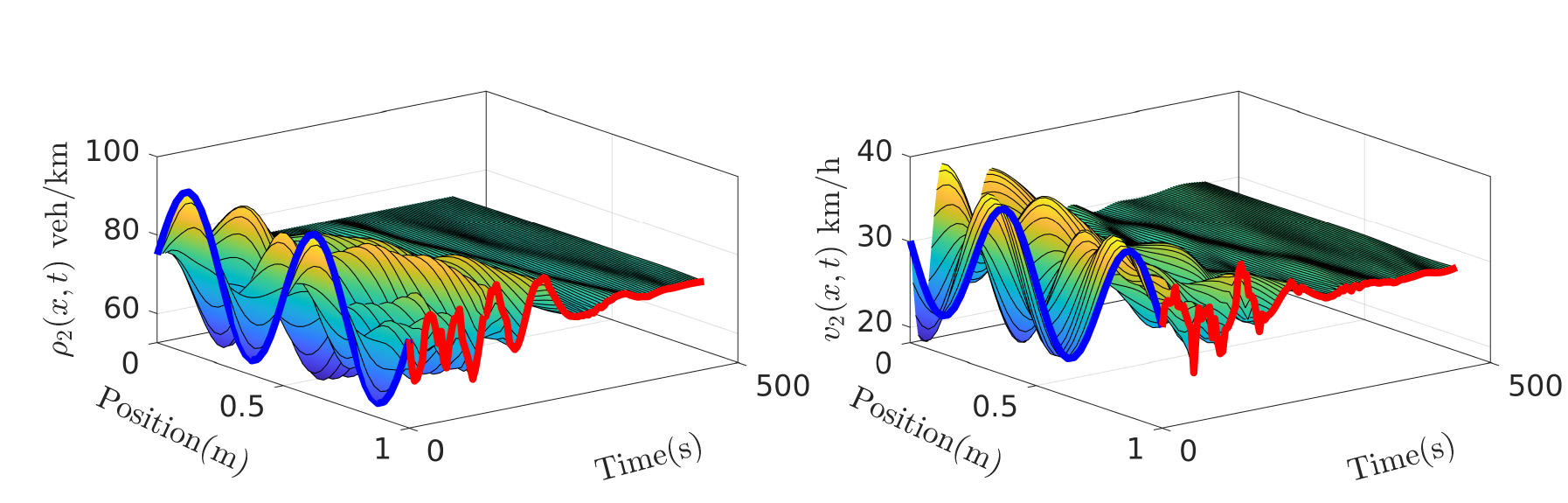}}
    \caption{The results under ETC with spacing $s_{\rm a}=16\text{m}$}
    \label{cl_event_s16}
\end{figure}
\begin{figure}[tbp]
    \centering
    \subcaptionbox{Control input \label{control_compare_s16}}{\includegraphics[width= 0.45\linewidth]{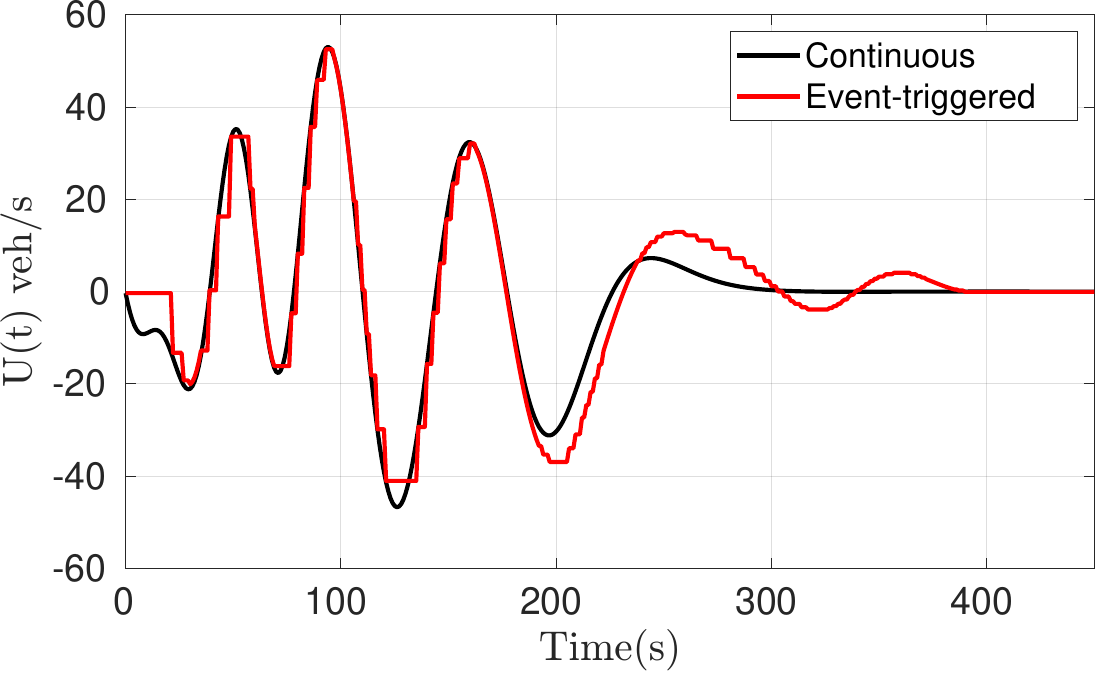}}
    \subcaptionbox{Release instants \label{releaseinterval_s16}}{\includegraphics[width= 0.45\linewidth]{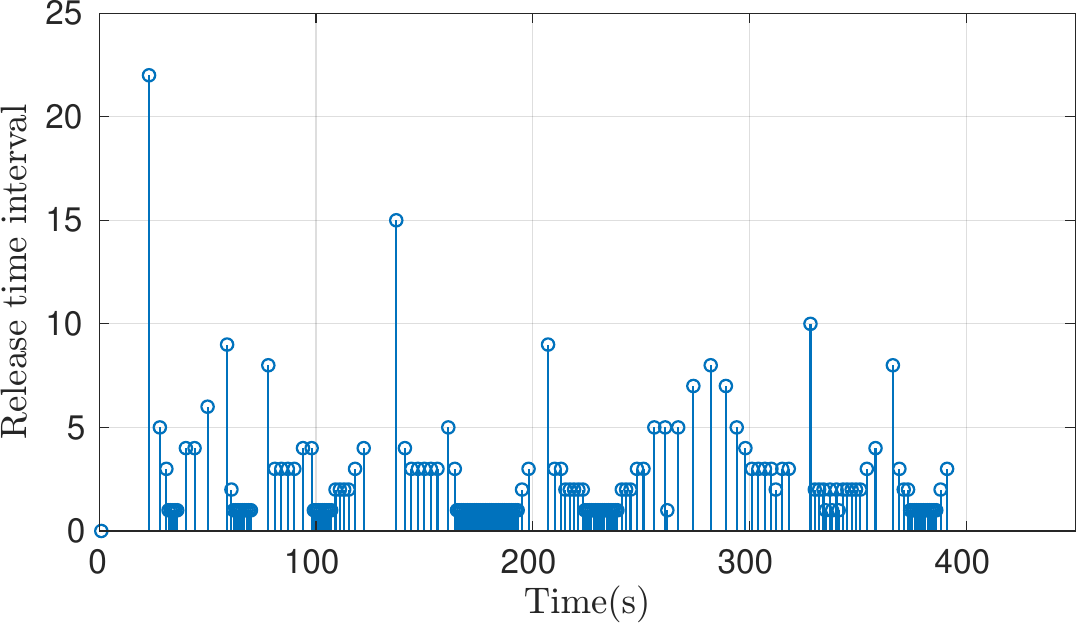}}
    \caption{The performance of ETC under $s_{\rm a}=16\text{m}$}
    \label{event-control_s16}
\end{figure}

\begin{figure}[htbp]
    \centering
    \subcaptionbox{Control input \label{control_compare}}{\includegraphics[width= 0.45\linewidth]{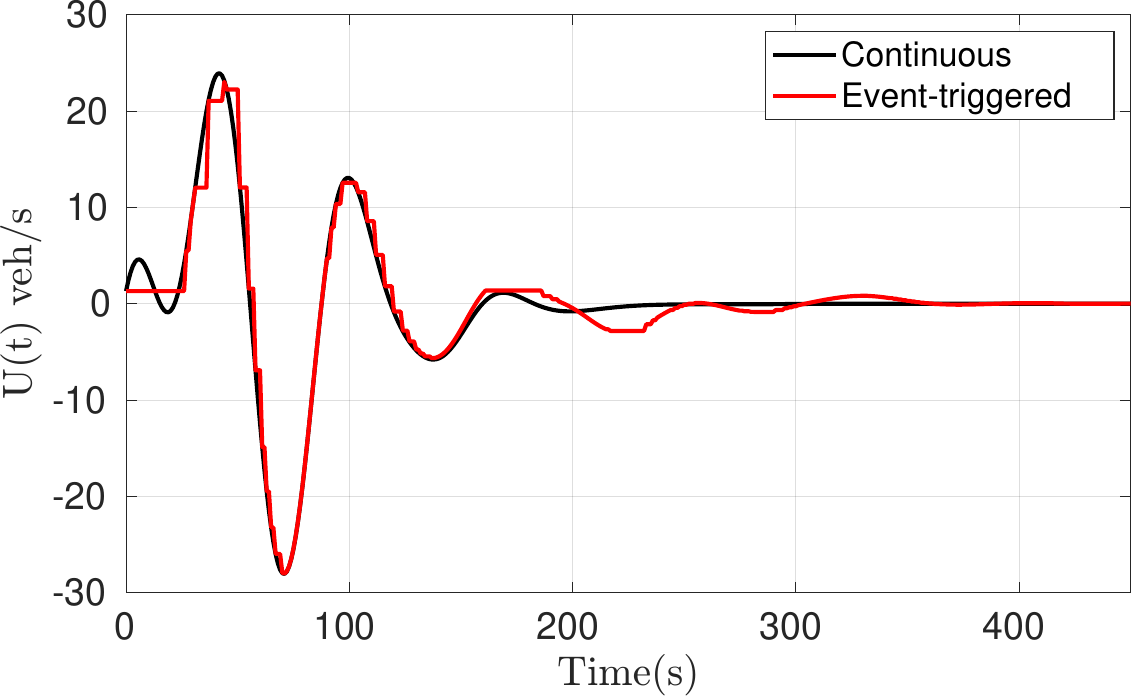}}
    \subcaptionbox{Release instants \label{releaseinterval}}{\includegraphics[width= 0.45\linewidth]{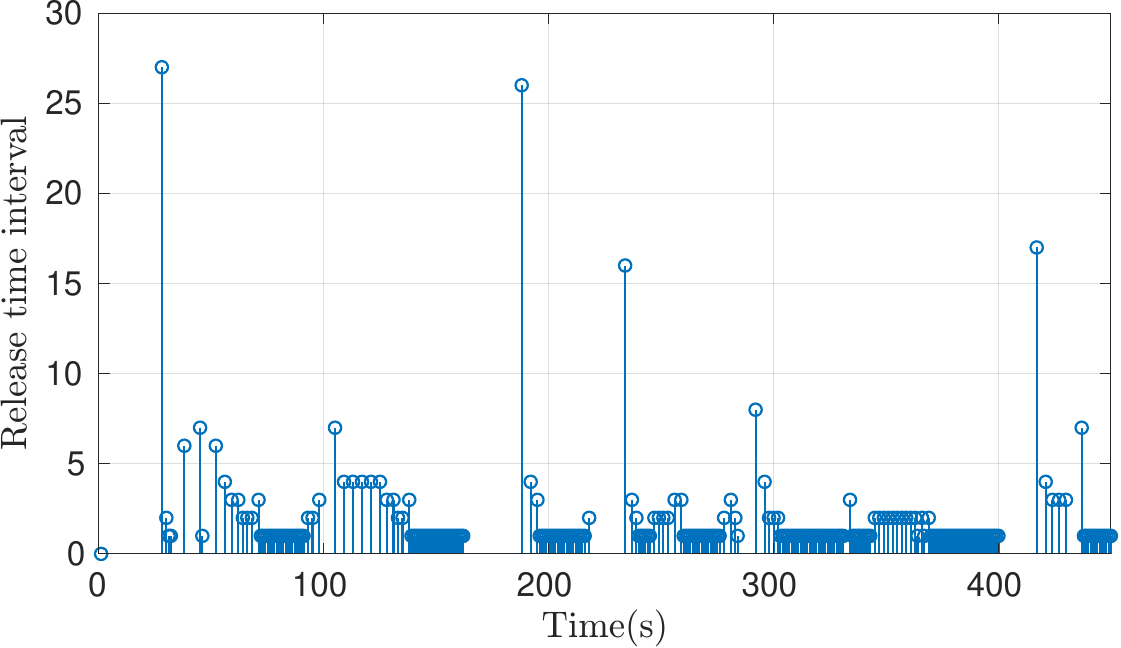}}
    \caption{The performance of ETC under $s_{\rm a}=20\text{m}$}
    \label{event-control}
\end{figure}
In this section, we provide the numerical simulation for the closed-loop system with event-triggered controller. Taking the equilibrium density as $\rho_{\rm h}^\star = 150\text{veh/km}$, $\rho_{\rm a}^\star = 75 \text{veh/km}$, such that $v_{\rm h}^\star = 29.16\text{km/h}$, $v_{\rm a}^\star = 13.32\text{veh/km}$ can be calculated by the fundamental diagram. The relaxation time is set as $\tau_{\rm h} = 30\text{s}$, $\tau_{\rm a} = 60 \text{s}$. The pressure exponent value is selected as $\gamma_{\rm h} = 2.5$, $\gamma_{\rm a} = 2$. The car-following gap are $s_{\rm h} = 5 \text{m}$, $s_{\rm a} = 16 \text{m}$. The maximum area occupancy $\overline{AO}_{\rm h} = 0.9$, $\overline{AO}_{\rm a} = 0.85$  In addition, we choose $\zeta = 8\times 10^{-3}$, $\sigma = 1\times10^{-4}$, $\eta = 0.9$, $A_1 = 2\times 10^{-2}$,  $A_2 = 3\times 10^{-3}$,  $A_3 = 4\times 10^{-3}$, $B = 9\times 10^{-3}$, and we also choose $\varrho_1 = 2\times 10^{-10}$, $\varrho_2 = 2\times 10^{-9}$, $\varrho_3 = 1.2\times 10^{-12}$, $\varsigma_4 = 0.01$, $\mu = 5\times 10^{-4}$. We run the simulation on a $L = 1000 \text{m}$ long road whose width is $6\text{m}$ and the simulation time is $450\text{s}$.

The closed-loop results of the mixed traffic system under the event-triggered controller are shown in Fig. \ref{cl_event_s16}. Fig. \ref{HVsevent_s16} represents the density and velocity of the HV while Fig. \ref{AVsevent_s16} denotes the AV. It can be found that the event-triggered controller stabilizes the mixed-autonomy traffic system. We also provide a comparison between the continuous controller using the backstepping method and the event-triggered controller in Fig. \ref{control_compare_s16}. The traffic management system does not need to update the control input at each time-step by using the event-triggered controller, therefore, the computational burden has been reduced. The triggered times and the release time interval are plotted in Fig. \ref{releaseinterval_s16}. 


We then test the different spacing settings of AV, we run the simulation under spacing of AV $s_{\rm a} = 20\text{m}$. The comparison between the continuous controller and the event-triggered controller is shown in Fig. \ref{control_compare}. The triggered times and the release time interval are plotted in Fig. \ref{releaseinterval}. 
The results show that the traffic system tends to become more congested with a larger space of the AV. The event triggered controller must need to execute more times to make the system stable.



\section{Conclusions}\label{sec6}


The event-triggered control of a mixed-autonomy traffic system is investigated. The traffic dynamics of the mixed-autonomy system are represented by an extended ARZ PDE model. A backstepping controller is designed to stabilize the system, and a dynamic ETM is defined. The event-triggered boundary controller is derived through this dynamic ETM. Lyapunov analysis is applied to show the stability results for the mixed-autonomy system with the event-triggered controller. Numerical simulations are conducted to illustrate the effects of the event-triggered controller. Future work will focus on developing an observer-based event-triggered controller for the mixed-autonomy traffic system.

\addtolength{\textheight}{-12cm}   





\bibliography{reference}

\end{document}